\documentclass[a4paper]{article}

\usepackage[english]{babel}
\usepackage[utf8x]{inputenc}
\usepackage[T1]{fontenc}

\usepackage[a4paper, top=0.3in, margin=0.95in]{geometry}

\usepackage{amsthm,amsmath,amsfonts,amssymb}
\usepackage{mathtools}
\usepackage{enumitem}
\usepackage{braket} 
\usepackage[colorlinks=true, allcolors=blue]{hyperref}
\usepackage{bbold} 
\usepackage[dvipsnames]{xcolor}
\usepackage{url}
\usepackage{authblk}

\theoremstyle{definition}

\theoremstyle{remark}

\newcommand{\R}{{\mathbb{R}}}
\newcommand{\T}{{\mathcal{T}}}

\newcommand{\W}{{\mathcal{W}}}
\renewcommand{\O}{{\mathcal{O}}}
\newcommand{\id}{{\mathbb{1}}} 

\newcommand{\A}{{A}}

\newcommand{\norm}[1]{\|{#1}\|}
\newcommand{\Norm}[1]{\left\|{#1}\right\|}
\newcommand{\Abs}[1]{\left|{#1}\right|}

\newcommand{\cc}[1]{\mathsf{{#1}}}	

\newtheorem{theorem}{Theorem}
\newtheorem{lemma}{Lemma}

\title{Implementing smooth functions of a Hermitian matrix on a quantum computer}
\author{\small Sathyawageeswar Subramanian$^1$\thanks{\texttt{ss2310@cam.ac.uk}} }
\author{\small Steve Brierley$^2$}
\author{\small Richard Jozsa$^1$}

\affil{\it \small $^1$ DAMTP, Centre for Mathematical Sciences \\ University of Cambridge, Cambridge~CB3 0WA, UK}
\affil{\it \small $^2$ River Lane Research, \\ 3 Charles Babbage Road, Cambridge, CB3 0GT, UK}
\date{}

\begin{document}
\maketitle


\begin{abstract}
We review existing methods for implementing smooth functions $f(A)$ of a sparse Hermitian matrix $A$ on a quantum computer, and analyse a further combination of these techniques which has some advantages of simplicity and resource consumption in some cases. Our construction uses the linear combination of unitaries method with Chebyshev polynomial approximations. The query complexity we obtain is $\O(\log C/\epsilon)$ where $\epsilon$ is the approximation precision, and $C>0$ is an upper bound on the magnitudes of the derivatives of the function $f$ over the domain of interest. The success probability depends on the $1$-norm of the Taylor series coefficients of $f$, the sparsity $d$ of the matrix, and inversely on the smallest singular value of the target matrix $f(A)$.

\end{abstract}


\section{Introduction}
There are many quantum algorithms that exhibit an advantage over known classical algorithms\footnote{at the time of writing there are over 380 papers listed in the quantum algorithms zoo}. Such a diversity of results can make it difficult to express the computational capability of a quantum computer in a clear and faithful manner to someone new to quantum computing. Here we address a large family of quantum algorithms that are often used as the main subroutine in many important applications. In particular we consider quantum algorithms that apply some smooth function of a Hermitian matrix to an input state. Examples of algorithms of this type include Hamiltonian simulation used in a variety of applications including quantum chemistry, the Quantum Linear Systems or matrix inversion algorithm \cite{Harrow2009QuantumEquations,Childs2015} used in quantum machine learning applications, and sampling from Gibbs distributions used in the recent quantum semi-definite programming (SDP) solvers \cite{Brandao2017ExponentialLearning,vanApeldoorn2017QuantumBounds}. 

In addition to briefly reviewing the methods for implementing matrix functions $f(A)$ used in these algorithms, we further use one of them to obtain an algorithm with query complexity expressed simply in terms of properties of the matrix $A$ and the Taylor expansion of $f$. Restricting to Hermitian matrices allows one to take advantage of their spectral decomposition to naturally extend results on approximating real functions to functions of matrices. Thus, for a real valued smooth function $f$, we look for a quantum algorithm which, when equipped with a quantum oracle for a Hermitian matrix $A$ and a map that prepares some state $\ket{x}$, returns a state that is close to $\frac{f(A)\ket{x}}{\Norm{f(A)\ket{x}}}$ in  $l_2$-norm. If measurements are involved, we require the desired output state to be obtained with high probability.

Among the earliest work of this kind, Klappenecker and R\"otteler \cite{Klappenecker2003EngineeringAlgorithms} studied the implementation of functions of unitary matrices, under some mild assumptions. This work embodies the same principles and ideas that were later developed into a method for Hermitian matrices in works of Berry, Childs, Kothari and collaborators. Other early work in this direction was motivated by studies on Hamiltonian simulation (for example, \cite{Somma2002SimulatingNetworks} and \cite{Childs2012HamiltonianOperations}), and algorithms for quantumly solving ordinary differential equations (for example, \cite{Leyton2008AEquations} describe a method to implement non-linear transformations of the amplitudes of a given input state). Kothari \cite{Kothari2014EfficientComplexity} gives a detailed description of probabilistic implementations of operators and technical lemmas on modified versions of amplitude amplification. Broadly, three different methods have emerged to realise the action of functions of Hermitian matrices on a quantum computer:

\begin{enumerate}
\item Using Hamiltonian simulation and Quantum Phase Estimation (QPE) as a technique to obtain a representation of the target state in the spectral basis of the matrix, followed by applying a suitably engineered unitary that computes the matrix function. This is the method used in the matrix inversion algorithm of \cite{Harrow2009QuantumEquations} and in the quantum recommendation systems algorithm of \cite{Kerenidis2016QuantumSystems}.

\item By representing the target matrix as a Linear Combination of Unitaries (LCU). Given an algorithm to implement each of the unitaries in the summation, the target matrix can be embedded in a unitary operation on a larger state space (adjoining the necessary ancillary qubits). This necessitates a post-selection step at the end, and so produces the required state with some associated success probability. This method has been used widely for Hamiltonian simulation and matrix inversion algorithms \cite{Childs2012HamiltonianOperations,Berry2015SimulatingSeries,Childs2015}.

\item More recently, Low and Chuang \cite{Low2016} have introduced a method called Qubitization and Quantum Signal Processing (QSP) (also referred to as the block-encoding method). They study how to implement functions of a diagonalisable complex matrix when provided with a unitary oracle and a signal state such that the action of the unitary on the subspace flagged by the signal state is proportional to the action of the target matrix. This method is thought to have optimal query complexity and optimal ancilla requirements for a large class of functions. This method has been substantially expanded and generalised to what has been termed Quantum Singular Value Transformation in a recent article of Gily\'en et.~al.~\cite{GSLW2018}.
\end{enumerate}

Typically, these methods have been applied to a specific function or application such as Hamiltonain simulation. The question of whether there exisits a more general result for an arbitrary function $f$ is very natural. Recently, van Apeldoorn et.~al.~\cite{vanApeldoorn2017QuantumBounds} used the Linear Combination of Unitaries or LCU method \cite{Childs2015} with an approximation by Fourier series to provide a constructive approach for the implementation of bounded smooth functions of a $d-$sparse Hermitian matrix (specified by an oracle) on quantum computers. They give an algorithm for constructing the approximating linear combination, and their quantum algorithm has query complexity linear in the sparsity $d$. This method uses Hamiltonian simulation as a black-box subroutine and involves converting the Taylor series of the function into a Fourier series, through a sequence of approximation steps. Both \cite{Low2016} and \cite{GSLW2018} also prove theorems about the classes of functions their methods can implement and describe the approach to construct the implementation for important examples like Hamiltonian simulation.

In this article, we first review the spectral method, the LCU method, and Qubitization in sections \ref{sec:HHL_method}, \ref{sec:QLSA-lcu}, and \ref{sec:qubitization} respectively. We then give an alternative constructive approach for implementing smooth functions of a Hermitian matrix on a quantum computer in section \ref{sec:LCU_Chebyshev}, based on the method of Chebyshev polynomial approximations used for matrix inversion in \cite{Childs2015}. The complexity of this alternative method is not directly comparable to that of \cite{vanApeldoorn2017QuantumBounds} but the main attraction is its simplicity. The use of Chebyshev polynomials allows the query complexity of the quantum algorithm to be directly obtained from the properties of the Taylor series expansion of $f$.


\section{Preliminaries}

We follow the usual model in which quantum algorithms access classical data (such as matrix entries) using a unitary oracle, and use the query complexity as a measure of efficiency. Another important quantity in the circuit model of quantum algorithms is the gate complexity, the number of $2$-qubit gates used. An algorithm is gate efficient if its gate complexity is larger than the query complexity by at most poly-logarithmic factors.  For details on the properties of the quantum walk construction including its gate complexity, we refer to \cite{Berry2012} and \cite{Berry2015HamiltonianParameters}.

A matrix is $d$-sparse if in any row, there are at most $d$ non-zero entries. For a $d$-sparse $N\times N$ matrix $\A$, we assume we have an oracle $\mathcal{P}_\A$ which performs two functions: 
\begin{eqnarray}
\label{eqn:Aoracle}
\ket{j, k, z} &\mapsto& \ket{j, k, z \oplus \A_{jk}} \\
\ket{j, l} &\mapsto& \ket{j, \text{col}(j, l)}
\end{eqnarray}
for all $j, k \in \{1,\ldots,N\}$, and $l \in \{1,\ldots,d\}$. The first line simply returns the matrix entry $A_{jk}$ in fixed-precision arithmetic. The second line computes the column index of the $l^{th}$ non-zero entry in row $j$, with the convention that it returns the column index of the first zero entry when there are fewer than $l$ non-zero entries in row $j$. When this can be efficiently done, the matrix is said to be efficiently row-computable.

Functions of a matrix are defined through its spectral decomposition. An $N\times N$ Hermitian matrix $A$ has $N$ real eigenvalues $\lambda_j$, with a spectral decomposition as a sum of projectors onto its eigenspaces spanned by the eigenvectors $\ket{u_j}$. A function $f$ of such a matrix is then defined as having the same eigenspaces, but with the eigenvalues $f(\lambda_j)$
\begin{equation}
\label{eqn:matrixFnDefn}
A=\sum_{i=1}^N \lambda_j \ket{u_j}\bra{u_j} \implies f(A)=\sum_{i=1}^N f(\lambda_j)\ket{u_j}\bra{u_j}.
\end{equation}
For finite dimensional matrices, any two functions $f$ and $g$ that are equal at the $N$ eigenvalues of the matrix give rise to the same matrix function. This suggests that it might be possible to treat matrix functions as corresponding \textit{interpolating} polynomials.


\section{The Spectral method}
\label{sec:HHL_method}
The first of the methods mentioned in the introduction uses Hamiltonian simulation and phase estimation as algorithmic primitives. The general framework is as follows.

For an $N\times N$ hermitian matrix $\A$, we first simulate Hamiltonian evolution under $\A$ for some time $t$ to prepare the state $e^{i\A t}\ket{\psi}$, for some initial state $\ket{\psi}$. The next step is to do Quantum Phase Estimation (QPE), logically decomposing the state in the spectral basis of $\A$, as a linear combination of product states with an ancillary register containing (approximate) eigenvalues of $\A$ normalised to $[0,1]$. At this stage the novelty lies in constructing an operation that will use the output state from the previous step to transform the probability amplitudes to the chosen function of the eigenvalues - typically a unitary conditioned on the register containing the eigenvalues. Schematically, 

\begin{equation*}
\ket{\psi} \xrightarrow{\text{H-sim}} e^{i\A t}\ket{\psi} \xrightarrow{\text{QPE}} \sum_{i=1}^N \psi_j \ket{a_j}\ket{\lambda_j} \xrightarrow{\text{Algo. (e.g. HHL) }} C \sum_{i=1}^N f(\lambda_j) \psi_j \ket{a_j}\ket{\lambda_j},
\end{equation*}
where $\A\ket{a_j}=\lambda_j\ket{a_j}$, $\ket{\psi}=\sum_{i=1}^N \psi_j \ket{a_j}$, and C is a normalisation factor. All three steps will have some associated error, and the last step will typically only succeed probabilistically. Choosing a precision $\epsilon>0$, the final state $\sum_{i=1}^N f(\lambda_j) \psi_j \ket{a_j}$ (having uncomputed the eigenvalue register coming from QPE) can be made $\epsilon$-close to $f(A)\ket{\psi}$, suitably normalised. 

Harrow, Hassidim and Lloyd \cite{Harrow2009QuantumEquations} used this method to devise an algorithm for solving a system of linear equations $A\vec{x}=\vec{b}$, in a formulation called the Quantum Linear Systems Problem (QLSP). The key step is matrix inversion, corresponding to the function $f(\lambda_j)=1/\lambda_j$. Given the oracle in (\ref{eqn:Aoracle}) for a Hermitian matrix $A$, and a state preparation map for an input vector $\vec{b}$, the HHL algorithm outputs a quantum state $\ket{\vec{x}}$ that is $\epsilon$-close in the $l_2$-norm to the normalised solution vector $A^{-1}\vec{b}/\norm{A^{-1}\vec{b}}$. The matrix inversion routine essentially implements conditional rotations to use phase-estimated eigenvalues to create the necessary amplitude factors. The transformation achieved by a simple conditional rotation is

\begin{equation*}
\ket{\mu}\ket{0} \xrightarrow{\text{CR}_y(2\cos^{-1}(\mu))} \ket{\mu} \left(\mu\ket{0}+\sqrt{1-\mu^2}\ket{1}\right),
\end{equation*}
where $\mu\in[0,1]$ and is represented to fixed precision in the first qubit register. 
Cao et.~al.~\cite{Cao2013QuantumEquation} use the HHL algorithm to solve the Poisson equation under some regularity assumptions, and give details of efficient quantum subroutines for the string of computations
\begin{equation*}
\ket{\lambda}\ket{0}\ket{0} \rightarrow \ket{\lambda}\ket{C/\lambda}\ket{0} \rightarrow
\ket{\lambda}\ket{C/\lambda}\ket{\cos^{-1}(C/\lambda)},
\end{equation*}
which are useful in matrix inversion. The actual function that is implemented is not quite $1/x$, but a carefully chosen filter function that matches with the inverse on the domain of interest, and ensures that the error can be kept under control. The choice of filter functions and their error analysis forms a major part of the work in designing the conditional unitary that implements the desired function.

\subsection*{Complexity and Hardness results for matrix inversion}
The HHL algorithm was originally presented as a method for the efficient estimation of averages or other statistical quantities associated with the probability distribution corresponding to the normalised solution vector of a linear system of equations. The advantage of the algorithm lies in being able to prepare this probability distribution as a quantum state, and the state may then be used to compute $\vec{x}^{\dagger}M\vec{x}:=\bra{x}M\ket{x}$ for operators $M$, or to sample from this distribution. 

When the matrix $A$ is $d$-sparse and efficiently row-computable, the algorithm has query complexity $\O(\textnormal{poly}(d, 1/\epsilon, \kappa))$. Here, $\epsilon$ is the accuracy to which the output state approximates the normalised solution vector, and $\kappa$ is the condition number of the matrix $A$. The condition number is given by the ratio of the largest to the smallest eigenvalues for a Hermitian matrix, and measures how invertible the matrix is: $\kappa\rightarrow\infty$ reflects an eigenvalue of the matrix going to zero, making it singular. The dependence of the algorithm on the condition number enters through the of filter function, which is chosen such that it matches $1/x$ on the domain $[-1,-1/\kappa]\cup[1/\kappa,1]$ and interpolates between these two pieces in $[-1/\kappa,1/\kappa]$. Consequently the success probability of the conditional rotation step also involves $\kappa$.

In the regime where the sparsity $d=\O(\log N)$, the HHL algorithm can be exponentially faster than the best known classical algorithms achieving the same results. It is still unknown whether better classical algorithms exist in the slightly modified framework of QLSP, i.e.,  the oracular setting where the goal is to compute global or statistical properties of the solution vector.

Nevertheless, the complexity of this algorithm was shown by \cite{Harrow2009QuantumEquations} to be nearly optimal (up to polynomial factors in the error) under the complexity theoretic assumptions $\cc{BQP}\neq\cc{PSPACE}$ and $\cc{BQP}\neq\cc{PP}$. The proof is an efficient reduction from simulating  a general quantum circuit to matrix inversion, establishing that matrix inversion is a $\cc{BQP}$-complete problem, so that the existence of a classical algorithm for this problem implies the ability to classically simulate quantum mechanics efficiently, which is widely believed to be impossible. The class $\cc{BQP}$ consists of problems that have a bounded error polynomial time quantum algorithm that succeeds with a constant probability $p\geq c > 1/2$, while $\cc{PP}$ consists of problems with probabilistic polynomial time classical algorithms with the success probability allowed to be arbitrarily close to $1/2$.

For the problem of preparing the state encoding the probability distribution corresponding to the solution vector, avoiding errors from the sampling step, Childs, Kothari and Somma \cite{Childs2015} designed an algorithm for QLSP with exponentially improved dependence on the precision parameter $\epsilon$. The key ingredient in their approach is the method of writing the target operator, $A^{-1}$ in this case, as a linear combination of unitary operators. We discuss this LCU method in the next section.


\section{The Linear Combination of Unitaries method}
\label{sec:QLSA-lcu}
One of the disadvantages in using QPE is that achieving $\epsilon$-precision requires $\O(1/\epsilon)$ uses of the matrix oracle. The LCU method offers a way to overcome this disadvantage by exploiting results from approximation theory. 

The LCU method is a way to probabilistically implement an operator specified as a linear combination of unitary operators with known implementations. In essence, we construct a larger unitary matrix of which the the matrix $f(A)$ is a sub-matrix or block. Childs and Wiebe \cite{Childs2012HamiltonianOperations} show how to implement a sum of two unitaries. We describe this simple case below.

Suppose $\A = \alpha_0 U_0 + \alpha_1 U_1$. Without loss of generality $\alpha_i>0$,  since phase factors can be absorbed into the unitaries. Consider a state preparation unitary $V_{\alpha}$ which has the action
\begin{align*}
\ket{0} &\mapsto \frac{1}{\sqrt{\alpha}}(\sqrt{\alpha_0}\ket{0} + \sqrt{\alpha_1}\ket{1}) \\
\ket{1} &\mapsto \frac{1}{\sqrt{\alpha}}(-\sqrt{\alpha_1}\ket{0} + \sqrt{\alpha_0}\ket{1}),
\end{align*}
where $\alpha=\alpha_0+\alpha_1$. When dealing with a linear combination of more than two unitaries, there is a lot of freedom in the choice of this $V_{\alpha}$, as we will see later.

Assume that we can perform the unitaries $U_0$ and $U_1$ controlled by an ancillary qubit, i.e., that we can apply the conditional unitary $U=\ket{0}\bra{0}\otimes U_0 + \ket{1}\bra{1}\otimes U_1$. Then for a state $\ket{\psi}$ for which we want $\A\ket{\psi}$, first attach an ancillary qubit and perform the map $V_{\alpha}$ on it, followed by $U$, and finally uncompute the ancilla with $V_{\alpha}^{\dagger}$. This results in the following transformation
\begin{eqnarray*}
\ket{0}\ket{\psi} &\xrightarrow{V_{\alpha}\otimes\id}& \frac{1}{\sqrt{\alpha}}(\sqrt{\alpha_0}\ket{0}+\sqrt{\alpha_1}\ket{1})\ket{\psi} \\
&\xrightarrow{U}& \frac{1}{\sqrt{\alpha}}(\sqrt{\alpha_0}\ket{0}U_0\ket{\psi}+\sqrt{\alpha_1}\ket{1}U_1\ket{\psi}) \\
&\xrightarrow{V_{\alpha}^{\dagger}\otimes\id}& \frac{1}{\alpha}\left(\ket{0}\left(\alpha_0U_0+\alpha_1U_1\right)\ket{\psi}+\sqrt{\alpha_0\alpha_1}\ket{1}\left(U_1-U_0\right)\ket{\psi}\right).
\end{eqnarray*}
Measuring the ancilla and getting the outcome $0$ will leave behind the state $\A\ket{\psi}$, up to normalisation. Getting a measurement outcome of $1$ means the algorithm fails. The probability of failure can be easily bounded, as $p(fail)\leq \frac{\alpha_1\alpha_0\Norm{U_1-U_0}^2}{\alpha^2}\leq \frac{4\alpha_1\alpha_0}{\alpha^2}$, since the distance between two unitaries is at most $2$.

\subsection*{Success probability and complexity}
In most cases of interest, the probability of success can be increased using amplitude amplification, by repeating the procedure $\O(\alpha/\Norm{\A\ket{\psi}})$ times, when we have an estimate of the norm in the denominator.
This gives a quantum algorithm that prepares the desired quantum state with constant success probability and outputs a single bit indicating whether it was successful. Indeed, when the LCU method is used to implement non-unitary operations such as in matrix inversion, a significant contribution to the complexity comes from the fact that usually the success probability is small. 

This framework was investigated in detail by Berry, Childs, Kothari and coworkers. The resulting method of implementing linear combinations of unitaries makes it straightforward to translate results on approximating real functions into implementations for matrix functions: since Hermitian matrices have all real eigenvalues, we just need to find a good approximation to the real function $f(x)$. The overall query complexity of the algorithm will depend on the $1$-norm or weight $\alpha$ of the coefficients of the linear combination, the number of terms $m$, and the least eigenvalue of the matrix function $f(A)$. 
  
Thus, finding a new algorithm boils down to optimising the first two parameters, and getting good bounds on the eigenvalues of $f(A)$.
We also need to choose the basis functions used in the approximation in such a way that on translating the statement to matrix functions, we get unitaries which we know how to implement - for example, a fourier series in $e^{ixt}$ gives rise to unitaries $e^{iAt}$ that can be implemented via Hamiltonian simulation techniques.
 
 \cite{Childs2015} used this method for the function $f(x)=1/x$ to obtain improved error dependence of the algorithm for the QLSP. The dependence of the complexity of their algorithm on the precision parameter is $\O(\log(1/\epsilon))$, an exponential improvement over the precision dependence in \cite{Harrow2009QuantumEquations}. This is achieved by carefully choosing the approximating series and its truncation, such that with $m=\O(\log(1/\epsilon))$ terms, $\epsilon$ precision is achieved. 

For completeness, a full description of the LCU method is included in Appendix \ref{app:LCU}. We now turn to the most recent method of implementing matrix functions, the method of qubitization and quantum signal processing, before returning to the LCU method in section \ref{sec:LCU_Chebyshev}.


\section{Qubitization or the Block-encoding method}
\label{sec:qubitization}

Introduced by Low and Chuang \cite{Low2016,Low2017OptimalProcessing}, this framework subsumes and generalises the LCU method. The method can be split into two steps - qubitization and quantum signal processing. The idea behind quantum signal processing is to ask what kinds of $2\times 2$ (block) unitaries can be obtained by iterating a given unitary operator, interleaving the iteration with single qubit rotations through different angles. Thus, picking a sequence of phases $\phi_0,\phi_1,\ldots,\phi_k$, we ask what class of untaries can be represented in the form 
\[
W_{\vec{\phi}}=e^{i\phi_0\sigma_z}U e^{i\phi_1\sigma_z}U\ldots Ue^{i\phi_k\sigma_z},
\]
where $U$ is the input unitary. When the input $U=U(x)$ is parametrised by some $x\in[-1,1]$, we can enforce $W_{\vec{\phi}}$ to have the form
\[
W_{\vec{\phi}}=\begin{pmatrix}
P(x) & Q(x) \\
Q^*(x) & P(x) 
\end{pmatrix}
\]
for some functions $P(x)$ and $Q(x)$, and work out what properties these functions can have. We can also calculate what choice of phases $\phi_j$ give rise to a certain $P$ and $Q$. This then provides a method for constructing functions of $U$ using iteration alternating with a sequence of single qubit rotations.

Qubitization, as the name suggests, is a technique of obtaining a suitable block encoding of an input matrix $A$, on which quantum signal processing can then be applied. The input is an $(m+n)$-qubit unitary matrix $U$ that encodes an $n$-qubit normal operator $A$ (i.e. $AA^{\dagger}=A^{\dagger}A$) in its top left block. More precisely, given a signal state $\ket{G}=\hat{G}\ket{0}$ that flags a subspace of the $n$-dimensional signal space, the input unitary $U$ has the block form
\begin{equation*}
U=\begin{pmatrix}
A & \cdot \\
\cdot & \cdot
\end{pmatrix},
\end{equation*}
i.e. $\braket{G|U|G}=A$ is the encoded operator (where we assume $A$ is normalised so that $\Norm{A}\leq 1$). Since normal operators have a spectral decomposition, we can write $A=\sum_{\lambda}\lambda e^{i\theta_{\lambda}}$. The goal is to use $U,G$, their controlled versions and conjugates to obtain a unitary $W$ that can be expressed as a direct sum over $SU(2)$ invariant subspaces, i.e., to `qubitize' $U$:
\begin{equation*}
W=\begin{pmatrix}
A & -g(A) \\
g(A) & A^{\dagger}
\end{pmatrix}
	= \sum_{\lambda} \begin{pmatrix}
\lambda e^{i\theta_{\lambda}} & -\sqrt{1-|\lambda|^2} \\
\sqrt{1-|\lambda|^2} & \lambda e^{-i\theta_{\lambda}}
\end{pmatrix}\otimes \ket{\lambda}\bra{\lambda}.
\end{equation*}
The unitary $W$, called an iterate, has an $SU(2)$ invariant subspace that contains the signal state $\ket{G}$, so that $\braket{G|W|G}=A$. This iterate can be used to approximately implement a wide range of matrix functions in the form $\braket{G|W_{\vec{\phi}}|G}=f(A)+ig(A)$, where $f,g$ are real functions, and $\vec{\phi}$ represents the sequence of phases $\phi_j$. The construction corresponds to quantum signal processing when the functions $f$ and $g$ have opposite parity, in which case the matrix function is implemented as a $2\times 2$ block unitary. Furthermore, the algorithms using this technique generally achieve polylogarithmic dependence on precision, and are shown to have optimal query complexity in several cases. The ancillary qubit requirement is brought down significantly since the primitive gates used are single qubit rotations, and controlled versions of the input unitaries. The iterate $W$ is similar to the quantum walk operator used with the Chebyshev decompositions for the LCU method. Indeed, when $A$ is hermitian, $W$ is exactly the same as that walk operator. 

The optimal Hamiltonian simulation algorithm based on the qubitization method \cite{Low2017OptimalProcessing} can be used to speed up any algorithm that uses Hamiltonian simulation as a subroutine.  Chakraborty, Gily\'en and Jeffery \cite{Chakraborty2018TheSimulation} have applied the qubitization method, also called the block encoding method, to obtain an improved Hamiltonian simulation technique for non-sparse matrices, improved quantum algorithms for linear algebra applications such as regression, and for estimating quantities like effective resistance in networks. Furthermore, they note that results in the block-encoding framework apply also when the input is specified in the quantum data structure model (e.g. QRAM) rather than as a unitary oracle. This connection is established through a result of Kerenedis and Prakash \cite{KP2017}, which shows that if a matrix $A$ is stored in a quantum data structure such as QRAM, an approximate block encoding for it can be implemented with polylogarithmic overhead.

The spectral and LCU methods were originally developed for and apply to Hermitian matrices, taking advantage of the fact that they have a spectral decomposition with real eigenvalues. Qubitization and quantum signal processing can deal with normal operators as well, which have complex eigenvalues. Very recently, \cite{GSLW2018} have further generalised the qubitization method using the idea of singular value decompositions. They show how matrix arithmetic and a variety of linear algebra applications can be achieved in a unifying framework of what they call `quantum singular value transformation'. Given any rectangular matrix, one has a singular value decomposition $A=U\Sigma V^{\dagger}$, where $\Sigma$ is the diagonal matrix of singular values, and $U$ and $V$ are orthogonal matrices with columns being the left and right eigenvectors of $A$ respectively. Essentially, \cite{GSLW2018} construct a method to implement functions of the singular values, i.e., functions defined by $f(A):=Uf(\Sigma)V^{\dagger}$. Their extensive analysis also shows how the method achieves optimal complexity for a wide variety of quantum algorithms.

In the last three sections we have briefly looked at the methods used for implementing smooth functions of (hermitian) matrices on a quantum computer. In the rest of this article, we analyse a certain combination of these methods, using the LCU technique with Chebyshev polynomial approximations and a quantum walk construction. 


\section{Implementing smooth functions using the LCU method}
\label{sec:LCU_Chebyshev}

Given a hermitian matrix $A$ and a smooth function $f:[-1,1] \to \mathbb{R}$, we describe below a quantum algorithm that implements the evolution proportional to $f(A)$ for any input state $\ket{\psi}$. 

\begin{theorem}
Consider a quantum oracle for a $d$-sparse hermitian matrix $A$ acting on $n$-qubits with $\Norm{A}\leq 1$, and a function $f:[-1,1] \to \mathbb{R}$. If $f$ has a Taylor series $f(x)=\sum_{}\alpha_i x^i$ about $x_0=0$, then for any $n$-qubit state $\ket{\psi}$, there is a quantum circuit that probabilistically prepares a state $\ket{\tilde{\psi}}=\left(\bra{0^t}\otimes\id\right)U\ket{0^t}\ket{\psi}$, where $U$ is a unitary acting on $t$ ancilla and $n$ system qubits, such that 
\begin{equation}
\Norm{\ket{\tilde{\psi}}-\frac{f(A)\ket{\psi}}{\Norm{f(A)\ket{\psi}}}}<\epsilon
\end{equation}
using $\O(L)$ queries to the oracle for $A$, with a success probability $p\geq \mu/\gamma$, where
\begin{equation}
L>\log\frac{C}{\epsilon},~~~\mu=\inf_x|f(x)|,~~~and~~~\gamma=\sum_{i=1}^Ld^i|\alpha_i|.
\end{equation}
The circuit also outputs a flag indicating success. 
\end{theorem}

Here $C>0$ is an upper bound on the magnitudes of the derivatives of $f$ in $(-1,1)$, and $\mu$ is the eigenvalue of $f(A)$ with the least magnitude on the domain of interest (i.e. on the spectrum of $A$).

For simplicity, we focus on the case where $\Norm{A}\leq 1$. The case when $1<\Norm{A}\leq\Lambda$ (and correspondingly the domain of $f$ is $[-\Lambda,\Lambda]$) can be reduced to this case by scaling appropriately. 

\subsection{Chebyshev series}
The circuit in Theorem $1$ is obtained using the LCU method. The first step is to derive an approximation for $f(A)$ in terms of Chebyshev polynomials. The matrix functions corresponding to these Chebyshev polynomials are then performed on a quantum computer using a quantum random walk. 

The Chebyshev polynomials of the first kind, denoted $T_n$ where $n$ is the degree, are orthonormal polynomials on $[-1,1]$ with weight function $(\sqrt{1-x^2})^{-1}$. We collect a few relevant facts about these polynomials in Appendix \ref{app:cheb}. In particular, we use two nice properties of Chebyshev polynomials in the quantum algorithm. The first is that monomials $x^k$ on $[-1,1]$ can be exactly represented as a finite sum of Chebyshev polynomials
\begin{equation}
\label{eqn:chebxk}
x^k=\sum_{j=0}^{k}C_{kj}\T_j(x),
\end{equation}
where the coefficents are given by
\begin{align}
\label{eqn:Ckj}
C_{kj}=\begin{cases}
      \displaystyle\frac{1}{2^{k-1}}\binom{k}{(k-j)/2}, &\text{if $(k-j)$ is even} \\
      0, &\text{otherwise},
      \end{cases}
\end{align}
for $j>0$, and $C_{k0}=\frac{1}{2^k}\binom{k}{k/2}$ is non-zero when $k$ is even. The second useful property is that the coefficients are positive and sum up to $1$
\begin{equation}
\label{eqn:chebCoefSum}
\sum_{j=0}^{k}C_{kj}=1,
\end{equation}
which can be seen using $\T_n(\cos \theta)=\cos n\theta$. 
We use these properties to write down a truncated Chebyshev series for $f(x)$, based on the Talyor series, which will lead to a simple expression for the success probability in the LCU method. 

For a smooth function $f$ on the interval $[-1,1]$, consider the Taylor series $x_0=0$, (also called the Maclaurin series)
\begin{equation}
\label{eqn:taylor_f}
f(x)=\sum_{i=0}^{\infty} \alpha_i x^i,
\end{equation}
where $\alpha_i=\frac{f^{(i)}(0)}{i!}$ is the $i^{th}$ Taylor coefficient, $f^{(i)}$ denoting the $i^{th}$ derivative of $f$. Suppose the radius of convergence of the series is some $r>0$. Truncating this series for some finite integer $L$, we get the Taylor polynomial, with truncation error bounded by Taylor's theorem
\begin{align}
\label{eqn:taylor_ftilde}
\tilde{f}&(x) = \sum_{i=0}^{L-1} \alpha_i x^i, \\
\forall x \in (-1,1),\ &\Abs{f(x) -\tilde{f}(x)} \leq \frac{f^{(L)}(\xi_L)}{(L)!}|\xi_L|^L,
\end{align}
for some $\xi_L \in (-1,1)$. Let us denote by $\alpha:=\sum_{i=0}^{L-1}|\alpha_i|$ the 1-norm of the coefficients. If we have bounds on the derivatives of $f$, or if we can bound the tail of the series as for the exponential function, we can get a good bound on the truncation error. This will enable us to quantify the rate of convergence of the series, and to decide the order of truncation based on the desired precision. 

A simple assumption that we can make about the derivatives of $f$ is that they are bounded in magnitude by some known constant $C$, i.e. $\sup_x|f^{(i)}(x)|\leq C$ for all $i$ and  $\forall x\in(-1,1)$. This holds for Schwartz functions, for example.  It is then a simple calculation to show that taking $L > \log_2 (C/\epsilon)$ ensures that the truncation error will be smaller than $\epsilon$, as long as $L>2e$, i.e. $L\geq 6$. \footnote{If we instead assume the weaker condition that $\sum_{i=0}^{\infty} |\alpha_i|<B$, and that the series converges in a $(1-\delta)$-ball around $x_0$, the corresponding truncation order is $L>\frac{1}{\delta}\log\frac{B}{\epsilon}$.}.

Now using (\ref{eqn:chebxk}) to represent the monomials $x^i$ exactly as a finite sum of Chebyshev polynomials, we obtain a truncated Chebyshev series for $f$. We have
\begin{align}
\label{eqn:cheb_f}
\tilde{f}(x) &=\sum_{i=0}^{L-1}\sum_{j=0}^i \alpha_i C_{ij}\T_j(x) \nonumber\\
			&=\sum_{j=0}^{L-1} \beta_j \T_j(x),
\end{align}
with $\beta_j = \sum_{i=j}^{L-1} \alpha_i C_{ij}$. From the observation (\ref{eqn:chebCoefSum}) above $\sum_{i=0}^k C_{ik}=1$, so
${\bf ||\beta||_1}:=\beta=\sum_{i=0}^{L-1} |\beta_i| = \sum_{i=0}^{L-1} |\alpha_i|$. The probability of success in the LCU implementation depends on this sum of coefficients, and we note that rewriting the Taylor polynomial as a Chebyshev decomposition does not by itself increase the weight of the coefficients.

However, taking $n$ steps of the quantum walk described in Appendix \ref{app:walk} results in the transformation $\ket{0^m}\ket{\psi}\mapsto \ket{0^m}\T_n(A/d)\ket{\psi}+\ket{\Phi^{\perp}}$. That is, the quantum walk implements the operator $\T_n(A/d)$ rather than $\T_n(A)$. To account for this, we further rewrite the series (\ref{eqn:cheb_f}) as
\begin{align}
\label{eqn:cheb_f_d}
\tilde{f}(x) &= \sum_{i=0}^{L-1} d^i \alpha_i \cdot\left(\frac{x}{d}\right)^i 
			\nonumber \\
			&=\sum_{j=0}^{L-1} \gamma_j \T_j\left(\frac{x}{ d}\right),
\end{align}
with $\gamma_j = \sum_{i=j}^{L-1} d^i \alpha_i C_{ij}$. This results in an increase in the weight of the coefficients: ${\bf ||\gamma||_1}:=\gamma=\sum_{i=0}^{L-1} |\gamma_i| = \sum_{i=0}^{L-1} | d^i \alpha_i|$. The truncation error does not change since we are expanding around $x_0=0$ and  simply rescaling the coefficients and argument of the series.

Finally, when  $1<\Norm{A}\leq \Lambda$, we can scale down the interval $[-\Lambda,\Lambda]$ to $[-1,1]$ using the map $x\mapsto x/\Lambda$. Hence the first step is to write the Taylor series in the larger interval, and then rewrite it by scaling the coefficients as in $\tilde{f}(x)=\sum_{i=0}^{L-1} \Lambda^i f_i \cdot (x/\Lambda)^i$. Accordingly, the weight or $1$-norm of the coefficients will increase to $\gamma=\sum_{i=0}^{L-1} |(\Lambda d)^i \alpha_i|$.

\subsection{Algorithm description and Complexity}
By Lemma $\ref{lem:CDapprox}$ (Appendix \ref{app:LCU}), we can implement $f(A)$ approximately by using the linear combination of Chebyshev polynomials in Eq. (\ref{eqn:cheb_f_d}), using the LCU method as described in Lemma \ref{lem:LCU_fApprox} (Appendix \ref{app:LCU}). We thus have a quantum algorithm which for an input state $\ket{\psi}$ produces a state $\ket{\tilde{\phi}}$ such that 
\begin{equation}
\Norm{\frac{f(A)\ket{\psi}}{\Norm{f(A)\ket{\psi}}}-\ket{\tilde{\phi}}} \leq c'\epsilon,
\end{equation}
for some constant $c'=\Omega(\frac{1}{\mu})$. The algorithm uses $\O(L)=\O(\log\frac{C}{\epsilon})$ queries to the matrix oracle, and succeeds with probability $p:=\left|\frac{\Norm{f(A)\ket{\psi}}}{\gamma}\right|^2\geq \left(\frac{\mu}{\gamma}\right)^2$, outputting the flag $0$ on success. Here $\mu$ is the eigenvalue of $f(A)$ with the least magnitude on the domain of interest.

Typically, amplitude amplification is used to boost the probability of success to a constant. The simplest setting where this is possible is when a state preparation map for for the input state $\ket{\psi}$ is available. Using Lemma \ref{lem:LCU_fApprox}, an upper bound on the worst-case query complexity of this implementation when amplitude amplification is used to boost the probability of success is given by 
\begin{equation*}
\frac{L}{\sqrt{p}} \leq L\frac{\gamma}{||f(A)\ket{\psi}||_{min}}=\O\left(L(\Lambda d)^{L-1} \frac{\alpha}{\mu}\right),
\end{equation*}
where $\Norm{A}\leq \Lambda$, and $\gamma:=\sum_{i=0}^{L-1} |(\Lambda d)^i f_i|=\O((\Lambda d)^{L-1}\alpha)$ and $\Norm{f(A)\ket{\psi}}\geq \mu$. In fact, $\gamma\approx f(\Lambda d)$. The linear factor of $L$ comes from the fact that we need to implement the Chebyshev polynomials of degree up to $L$ using the quantum walk. The factor $\gamma/\mu$ comes from using amplitude amplification to increase the success probability of obtaining the desired state. A simple lower bound on $\mu$ is $f_{min}:=\inf_{x}|f(x)|\leq \mu\leq \Norm{f(A)\ket{\psi}}$ over $x\in[-1,1]$. Thus if $f$ is such that $|f(x)|\geq 1$, this factor can be omitted from the complexity. The implementation is invariably expensive when $f(A)$ has eigenvalues close to zero.  

Thus, if amplitude amplification is used, plugging in the expression for $L$ gives
\begin{equation}
\label{eqn:comp1}
\O\left(\frac{\alpha}{\mu}\left(\frac{C}{\epsilon}\right)^{\log(\Lambda d)}\log\frac{C}{\epsilon}\right)
\end{equation} 
queries to the oracle for $A$, and $\O\left(\frac{\alpha}{\mu}\left(\frac{C}{\epsilon}\right)^{\log(\Lambda d)}\right)$ uses of the input state preparation map.

The gate complexity can be obtained by multiplying the query complexity by the gate complexity of performing one step of the quantum walk. From \cite{Berry2015HamiltonianParameters}, one step of the walk costs only a constant number of queries and $\O\left(\log N+m^{2.5}\right)$ $2$-qubit gates, where $m$ is the number of bits of precision used for the entries of the matrix $A$. 
For completeness, note that the control state preparation map $V$ in the LCU method can be constructed using the method of \cite{Grover2002CreatingDistributions}, since the coefficients are known.


\section*{Related work}
As noted previously, \cite{vanApeldoorn2017QuantumBounds} provide a constructive method for implementing smooth functions of Hermitian matrices, based on transforming the Taylor series for the function into a Fourier series. The algorithm is then obtained by using the LCU method and Hamiltonian simulation to implement the Fourier components. For comparison, we quote below the results as described in Theorem $40$ of their paper. 

Given the Taylor series of $f$ about some point $x_0$, $f(x_0+x)=\sum_{i=0}^{\infty}a_ix^i$, with convergence radius $r>0$, if an $n$-qubit Hermitian operator $A$ satisfies $\Norm{A-x_0\id}\leq r$, then for $\epsilon\in(0,\frac{1}{2}]$ we can implement a unitary $\tilde{U}_f$ on $t+n$-qubits such that for any $n$-qubit state $\ket{\psi}$, we have
\begin{equation*}
\displaystyle\Norm{\left(\bra{0^t}\otimes\id\right)\tilde{U}_f\ket{0^t}\ket{\psi}-\frac{f(A)}{B}\ket{\psi}} \leq \epsilon,
\end{equation*}
where $\sum_{i=0}^{\infty}|a_i|(r+\delta)^i\leq B$ for some finite $B>0$ and $\delta\in(0,r]$. 
If $\Norm{A}\leq K$, $r=\O(K)$, and $A$ is $d$-sparse and accessible using an oracle as in (\ref{eqn:Aoracle}), then the whole circuit can be implemented using 
\[
\O\left(\frac{Kd}{\delta}\log\left(\frac{K}{\delta\epsilon}\right)\log\left(\frac{1}{\epsilon}\right)\right)
\] queries to the oracle. The number of $3$-qubit gates used for the circuit is larger by the polylogarithmic factor $\O\left(\log N+\log^{2.5}\left(\frac{K}{\delta\epsilon}\right)\right)$. 

We notice that the query complexity in this method depends linearly on $d||A||$, in addition to the factor $\log 1/\epsilon$. In the Chebyshev method, the dependence on $d$ and $||A||$ comes in through the success probability, leaving the query complexity dependent only on the properties of the Taylor series approximation. 

The methods based on Fourier and Chebyshev series cannot be directly compared because they may not be usable in all situations - the Fourier method is possible whenever Hamiltonian simulation can be performed for the matrix $A$, while the Chebyshev method is possible only when the quantum walk in Appendix \ref{app:walk} is feasible to implement. Since Hamiltonian simulation can be performed using quantum walks, the Fourier method has a wider range of applicability, in general. 

The probability of success in preparing the desired state by postselection on the $t$-ancillary qubits is as expected in the LCU method, given by $\left|\frac{\Norm{f(A)\ket{\psi}}}{B}\right|^2$.


\section*{Advantages of the Chebyshev method}
In comparison to the method based on Fourier series approximations, the Chebyshev series based method described here has the following advantages.
\begin{enumerate}
\item Since the quantum walk exactly produces the effect of Chebyshev polynomials in $A/d$ (apart from a choice of precision in representing the entries of the matrix), and polynomials have exact representations as a finite sum of Chebyshev polynomials, we can exactly implement polynomial functions of a matrix. In the Fourier series based approach, the series truncation adds another layer to the error in the approximation of polynomials. The simplicity of the analysis also makes it apparent that the Chebyshev method could be more suitable for applications involving polynomial functions, such as iterative methods that use Krylov subspaces.

\item Leaving out amplitude amplification, the query complexity depends only on the degree of the approximating polynomial. This isolation of the dependence of the complexity on the norm of the matrix and its sparsity into the success probability could be useful in studying lower bounds on the query complexity for different matrix functions.

\item Methods based on quantum walks extend to non-sparse matrices, since they do not depend on row computability \cite{Berry2012}. The complexity will generally be worse, however. 

\item The classical calculation of the coefficients is particularly simple for the Chebyshev series.
\end{enumerate}

There are, of course, both advantages and disadvantages of using any method. Chebyshev polynomial implementation uses quantum walk methods, and the construction of the walk requires a doubling of the input space, i.e., $\O(n)$ ancillary qubits. Furthermore, the rescaling of the Taylor series coefficients means that machine errors due to fixed-precision representation are magnified. To work at precision $\epsilon$, we need to use $\Omega(\log_2(B/\epsilon))$ bits for the matrix entries.


\section{Special function classes}
The main difficulty in the approach we have described is the scaling up of the weight of the coefficients in the Taylor series approximation due to the fact that the Chebyshev polynomials obtained from the quantum walk are in $A/d$ rather than $A$. This affects the success probability, potentially necessitating many rounds of repetition or amplitude amplification. The step that requires this rescaling is the reconstruction of $f(x)$ using an approximation to $f(x/d)$ for a fixed $d>1$. 

But we notice that there are simple functions for which different approaches are possible. For example, for $f(x)=1/x$, it suffices to take $f(x)=\frac{1}{d}f(\frac{x}{d})$, which holds throughout the domain of $f$. Another example is $e^x = (e^{x/d})^d$, which requires repeated application of the operator, $d$ times. We also need to keep track of how the approximation error changes in going from $f(x/d)$ to $f(x)$. 

In general, if there is a function $g:\R \to \R$ such that $f(x)=g(f(x/d))$, the efficiency of implementing of $f(A)$ using $f(A/d)$ depends on the nature of $g$. For some classes of functions, the composition of the function with $g$ reduces to just scaling: $f(x)=g(f(x/d))=g(d)f(x/d)$. Homogeneous functions are an example of this kind: $f(cx)=c^kf(x)$ for a fixed constant $k$ for any real number $c$, so $g(d):=d^k$. Homogenous functions arise mainly as (multivariate) polynomials or rational functions. Below, we make some brief remarks about matrix polynomials and exponentiation.

\subsection*{Polynomials}
For the special case of monomials, which are homogeneous functions, we simply use the Chebyshev decomposition \eqref{eqn:chebxk}, which is exact. This eliminates the precision parameter $\epsilon$. For $f(x)=x^k$ on $[-1,1]$, the query complexity is $\O(k)$, or if amplitude amplification is used, $\O\left(kd^{k-1}\lambda^{-k}\right)$ where $\lambda=||A^{-1}||$ is the least eigenvalue of $A$ (where we assume $A$ does not have $0$ as an eigenvalue). For any polynomial of degree $k$, the complexity is the same to leading order, since the highest degree Chebyshev polynomial comes from the highest degree term in the polynomial. Computing matrix polynomials may find use in iterative methods in numerical linear algebra. These methods typically proceed by approximating a vector $f(A)\vec{v}$ in a Krylov subspace $\mathcal{K}_r(A,\vec{v_0}):=\set{\vec{v_0}, A\vec{v_0}, A^2\vec{v_0},\ldots,A^{r}\vec{v_0}}$ starting with an initial guess $\vec{v_0}$, and iteratively improving it. Patel and Priyadarsini \cite{Patel2017EfficientApplications} propose a matrix inversion algorithm using this method. However, they use a different formalism to quantumly implement the monomials. 

\subsection*{The exponential function}
The matrix exponential is an immensely important function, primarily in the form of the complex exponential $e^{iAt}$ that describes quantum evolution under the Hamiltonian operator $A$. Hamiltonian simulation is important in a variety of applications, ranging from quantum chemistry, to use as a subroutine in linear algebra and machine learning applications. This vast topic has received a lot of attention in the last two decades, so we shall not endeavour to elaborate on it here.

The simple exponential $e^A$ is also an important function. For example, being able to sample from the Gibbs' state $e^{-H}/\text{Tr}(e^{-H})$ has been found useful in quantum algorithms for semidefinite programming \cite{Brandao2017ExponentialLearning,VanApeldoorn2018ImprovementsApplications}. Exponentiating density matrices can be used to construct a quantum algorithm for principal component analysis \cite{Lloyd2014QuantumAnalysis}.  Matrix exponentiation is also expected to be useful in variational quantum chemistry algorithms, for example in implementing coupled cluster techniques \cite{Romero2017StrategiesAnsatz}. 

Using our approach, to implement the exponential $e^A$ of the hermitian matrix $A$, we can first approximate $e^{A/d}$ and repeat this operation $d$ times. If $||A||\leq 1$, this leads to a query complexity of $\O(d\log 1/\epsilon)$. However, the success probability decays exponentially and thus many rounds of amplitude amplification may be required. Following the construction in section \ref{sec:LCU_Chebyshev}, we note that $\gamma=e^{d}$ to order $\epsilon$, so the complexity of using amplitude amplification can only be constrained to $\O(e^{d})$. 

Many authors have considered the problem of matrix exponentiation previously. Some have considered the problem of solving a system of linear ordinary differential equations \cite{Berry2014High-orderEquations,Berry2017}, while others have focused on the problem of Gibbs' state preparation \cite{Chowdhury2016}. \cite{vanApeldoorn2017QuantumBounds} also describe an algorithm for this problem, which they use for Gibbs' sampling in quantum SDP algorithms. \cite{Patel2017EfficientApplications} give an algorithm that also uses the Chebyshev expansion of the exponential function, but their method of implementation is based on the recursion relation for Chebyshev polynomials, and uses a digital encoding of quantum states.


\section{Conclusions}
We have shown a simple calculation motivated by the approximation of real functions by Chebyshev series to illustrate the use of quantum walks with the LCU method to implement a wide variety of smooth functions. The method is particularly simple, as the approximating linear combination is obtained from a truncated Taylor series that is transformed into a Chebyshev series. Although we do not present any improvements in complexity results, the simplicity of the method may make it attractive from a pedagogical viewpoint. 

Although the LCU method has been widely investigated over the last few years, there are still a few interesting questions related to it - for example, the only unitaries found useful so far are tensor products of Pauli operators, Fourier basis terms, and Chebyshev polynomials, because these can be implemented using known methods (Hamiltonian simulation and quantum walks). It would be interesting to study other families of unitary circuits that can be used in conjunction with this method. In particular, for a special case such as say matrix exponentiation, is it possible to systematically determine an optimal basis of unitaries? This would also be related to lower bounds on the query complexity of implementing different matrix functions. 

\section*{Acknowledgements}
SS is supported by a Cambridge-India Ramanujan scholarship from the Cambridge Trust and the SERB (Govt. of India). RJ is supported in part by the ERANET cofund project QuantAlgo.


\newcommand{\etalchar}[1]{$^{#1}$}

\appendix
\section{Using the LCU method to approximate $f(\A)$ for a hermitian matrix $\A$}
\label{app:LCU}
We describe below the LCU method for approximately implementing a linear combination of unitary matrices. This description is based on \cite{Childs2015}.

First, we need to make sure that approximating a hermitian operator $C$ by another operator $D$ in spectral norm ensures that they produce states $C\ket{\psi}/\Norm{C\ket{\psi}}$ and $D\ket{\psi}/\Norm{D\ket{\psi}}$ that are close in the Hilbert space norm, for any state $\ket{\psi}$.

\begin{lemma}
\label{lem:CDapprox}
Let $C$ be a Hermitian operator whose eigenvalues satisfy $\Abs{\lambda}\geq 1$, and $D$ be an operator satisfying $\Norm{C-D}\leq\epsilon<1/2$. Then for any state $\ket{\psi}$, 
\begin{equation}
e(\psi):=\Norm{\frac{C\ket{\psi}}{\Norm{C\ket{\psi}}}-\frac{D\ket{\psi}}{\Norm{D\ket{\psi}}}}<4\epsilon.
\end{equation}
\end{lemma}

\begin{proof}
Since the assertion is about normalised states, we can consider states with $\norm{\ket{\psi}}=1$ without loss of generality. Repeated application of the triangle inequality and the fact that $\Norm{C\ket{\psi}}\geq\Abs{\lambda_{min}}$ together imply that
\begin{equation}
e(\psi)\leq \frac{2\epsilon}{\Abs{\lambda_{min}}-\epsilon}.
\end{equation}
We can replace this with a looser bound $e(\psi)<c\epsilon$ for some constant $c>0$; then $c$ must satisfy $\frac{2}{c}<\Abs{\lambda_{min}}-\epsilon$, i.e., $c=\Omega(\frac{1}{|\lambda_{min}|})$. This indicates that for operators with eigenvalues that approach zero in magnitude, the approximation in normalised states is worse, because we will need to settle for a larger $c$. 

If we assume for the eigenvalues of $C$, as stated in the lemma, that $\Abs{\lambda_{min}}\geq 1$, then using $\epsilon<1/2$ gives $c\geq 4$, so that $e(\psi)<4\epsilon$.
\end{proof}

We can now state precisely how a matrix function will be approximated by a finite linear combination of unitaries. 

\begin{lemma}{\cite[Corollary 10]{Childs2015}}
\label{lem:LCU_fApprox}
Let $\A$ be a hermitian operator acting on $n$-qubits with eigenvalues lying in an interval $D\subset \mathbb{R}$. Suppose the function $f:D\rightarrow\R$ can be approximated by the linear combination of functions $g_i:D\rightarrow\R$ for $i=1,\ldots,m$ such that 
\begin{equation}
\sup_{x \in D}\Abs{f(x)-\sum_{i=1}^m g_i(x)}\leq \epsilon.
\end{equation}
Further, let $\{U_i:i=1,\ldots,m\}$ be a set of unitaries on $(n+t)$-qubits (with $t=\O(\log m)$) that satisfy
\begin{equation}
U_i\ket{0^t}\ket{\psi}=\ket{0^t}g_i(\A)\ket{\psi}+\ket{\Psi_i^\perp},
\end{equation}
$\forall\ n$-qubit states $\ket{\psi}$, with $\left(\ket{0^t}\bra{0^t}\otimes\id_n\right)\ket{\Psi_i^\perp}=0$. Given an oracle $\mathcal{P}_{\vec{b}}$ that prepares a state $\ket{b}$, there is a quantum algorithm that prepares with high probability a state $\ket{\tilde{\psi}}$ such that
\begin{equation}
\Norm{\ket{\tilde{\psi}}-\ket{\psi}}\leq c\epsilon, ~~~\text{with}~~~\ket{\psi}=\frac{f(\A)\ket{b}}{\norm{f(\A)\ket{b}}}.
\end{equation}
The algorithm uses amplitude amplification, making $\O(\alpha/\norm{f(\A)\ket{\psi}})$ queries to $\mathcal{P}_{b}$, and to the following unitary operators
\begin{align}
U:=\sum_{i=0}^m\ket{i}&\bra{i}\otimes U_i,
\ \ \ V\ket{0^s}=\frac{1}{\sqrt{\alpha}}\sum_{i=0}^m\sqrt{\alpha_i}\ket{i} \nonumber\\
W&:=(V^\dagger \otimes \id_n )U(V\otimes\id_n).
\end{align}
\end{lemma}

\begin{proof}
By a straightforward calculation,
\begin{equation*}
W\ket{0^s}\ket{0^t}\ket{\psi}=\frac{1}{\alpha}\ket{0^s}\ket{0^t}f(\A)\ket{\psi}+\ket{\Psi^\perp}.
\end{equation*}
Since the state preparation map for $\ket{\psi}$ is available, standard amplitude amplification can be applied.
\end{proof}

\section{Amplitude Amplification}
\label{app:AA}
Given a classical input vector $\vec{b}=(b_1,\ldots,b_N)^T \in \mathbb{C}^N$, we usually assume we have an oracle $\mathcal{P}_{\vec{b}}$ that prepares a quantum state $\ket{\psi}$ corresponding to this vector, defined by 
\begin{equation}
\label{eqn:stateOracle}
\ket{0^m} \mapsto \frac{\sum_{i} b_i\ket{i}}{\norm{\sum_i b_i\ket{i}}} := \ket{\psi},
\end{equation}
where $m = \left\lceil \log N \right\rceil +1$. When such a unitary oracle is available, standard amplitude amplification (AA) can be applied. The initial state on which we perform AA is 
\begin{equation}
\ket{\Psi_i}:=W\ket{0^m}\ket{\psi}=\sqrt{p}\ket{0^m}\ket{\phi}+\sqrt{1-p}\ket{\Phi^{\perp}},
\end{equation}
where $\ket{\phi}=\frac{A\ket{\psi}}{\Norm{A\ket{\psi}}}$, $p=\Abs{\frac{\Norm{A\ket{\psi}}}{\alpha}}^2$ and $\ket{\Phi^{\perp}}$ is normalised. The projection onto the desired subspace is $\Pi:=\ket{0^m}\bra{0^m}\otimes\id_n$, which picks out the $0$-flag on the ancillary register. The reflection about the target subspace can then be taken as $R_t:=(\id_{m+n}-2\Pi)$, since it leaves the system register unchanged, while reflecting about the ancillary success flag. Then the reflection about the initial state is 
\begin{equation*}
R_i:=W(\id_m\otimes P_{\vec{b}})R_t(\id_m\otimes P_{\vec{b}})^{\dagger}W^{\dagger},
\end{equation*}
and the usual grover iterate is obtained, $G=-R_iR_t$.

Often such a state preparation map will not be available. In such cases, it may still be possible to use a version called `oblivious' amplitude amplification, if the matrix function we are attempting to implement is close to unitary. This is done using the Grover iterate like operator $S:=-WR_tW^{\dagger}R_t$. More details can be found in \cite{Berry2015SimulatingSeries} and \cite{Kothari2014EfficientComplexity}.

\section{Chebyshev polynomials}
\label{app:cheb}
It is a result in approximation theory that Chebyshev polynomials form the best polynomial basis for approximating functions on $[-1,1]$ in the supremum or $L_{\infty}$ norm. That is, they minimise the error $\sup_x|\tilde{f}(x)-f(x)|$ between the approximator $\tilde{f}$ and the target $f$.

The Chebyshev polynomials of the first kind, $\T_n(x)$, satisfy $\forall x \in [-1,1]$
\begin{enumerate}
\item $\T_0(x) = 1$, $\T_1(x) = x$ 
\item $\T_{n+1}(x) = 2x\T_n(x)-\T_{n-1}(x)$
\item  
\begin{align}
\int_{-1}^{1} \frac{\T_n(x) \T_m(x)}{\sqrt{1-x^2}} dx = 
    \begin{cases}
	0, &\text{$m \neq n$} \\
    \pi/2, &\text{$m=n \neq 0$} \\
    \pi, &\text{$m=n=0$}
	\end{cases}
\end{align}
\end{enumerate}

We can exactly represent the monomials $x^k$ on $[-1,1]$ in the basis of Chebyshev polynomials as the finite sum of terms up to degree $k$ as
\begin{equation}
\label{eqn:chebxk_appendix}
x^k=\sum_{j=1}^{k}C_{kj}\T_j(x) + \frac{1}{2}C_{k0}.
\end{equation}

The coefficients $C_{ij}$ can be calculated by the substitution $x=\cos(\theta)$, and using the property $\T_n(\cos(x))=\cos(nx)\ \ \forall x \in [0,\pi]$, as
\begin{align}
C_{kj}&=\frac{2}{\pi}\int_{-1}^1\frac{x^k\T_j(x)}{\sqrt{1-x^2}}dx \nonumber\\
      &=\frac{2}{\pi}\int_{0}^\pi \cos^k(x)\cos(jx)dx.
\end{align}

Writing $\cos(x)=\frac{1}{2}(e^x+e^{-x})$ and $\cos(jx)=Re(e^{ijx})$, we get
\begin{equation}
C_{kj}=\frac{1}{2^{k-1}\pi}Re\left(\int_{0}^\pi \sum_{l=0}^k \binom{k}{l} e^{i(2l-(k-j))}dx\right) .
\end{equation}

The real part of the exponential integrates to zero on $[0,\pi]$ unless $2l=k-j$. In this case, the integral is just $\pi$, the length of the interval, and we get
\begin{align}
C_{kj}=\begin{cases}
      \displaystyle\frac{1}{2^{k-1}}\binom{k}{(k-j)/2}, &\text{if $(k-j)$ is even} \\
      0, &\text{otherwise}.
      \end{cases}
\end{align}
This makes sense, since for $k$ even, $x^k$ is an even function and will contain only Chebyshev terms of even degree in its expansion (similarly when $k$ is odd).

The Chebyshev polynomials $\T$ all evaluate to $1$ at $x=1$ (can be seen from $\T_n(\cos(x)) = \cos(nx)$), and this gives us a useful property: for all integers $k\geq 0$
\begin{align}
\sum_{j=0}^{k}&C_{kj}=1.
\end{align}
For the quantum walk construction, we also need a few properties related to the Chebyshev polynomials of the second kind. We have that $\forall x \in [-1,1],~\T_0(x)=U_0(x)=1,~T_1(x)=x,~U_1(x)=2x$ and both $\T_n(x)$ and $U_n(x)$ satisfy the same recursion relation
\begin{equation}
f_{n+1}=2xf_{n}-f_{n-1},
\end{equation}
for all positive integers $n$, where $f$ represents either $\T$ or $U$. From this, it can be seen that $\T_n$ and $U_n$ satisfy the relations
\begin{eqnarray}
\lambda\T_{n-1}(\lambda)+(\lambda^2-1)U_{n-2}(\lambda)&=&\T_n(\lambda) \nonumber\\
\T_{n-1}(\lambda)+\lambda U_{n-2}(\lambda)&=&U_{n-1}(\lambda),
\end{eqnarray}
which turn out to be useful in the next section. Furthermore,
\begin{eqnarray}
\T_{n}(\cos(x))&=&\cos(nx) \nonumber\\
U_{n}(\cos(x))&=&\frac{\sin((n+1)x)}{\sin(x)},
\end{eqnarray}
where $x=\cos^{-1}(\lambda) \in [0,\pi] \ \forall \lambda \in [-1,1]$.

\section{Implementing Chebyshev polynomials in $\A$ using a quantum walk}
\label{app:walk}
To actually realise an implementation based on the LCU method, we need to be able to perform the unitary matrices in the decomposition of the target matrix. One of the families of unitaries that have been found useful is a quantum walk operator $\W$ which has the property that when restricted to a certain invariant subspace, $\W^n$ has a block form with the first block being the operator $\T_n(H)$, where $\T_n$ is the Chebyshev polynomial of the first kind and degree $n$, and $A=dH$ is the matrix using which the walk is constructed. We describe this quantum walk operator in this section for a normalised Hermitian matrix $A$. More details can be found in \cite{Childs2015}.

Given a $d$-sparse Hermitian matrix $\A$ acting on $\mathbb{C}^N$, we consider two copies of the system, each adjoined with a single ancillary qubit to form $\mathbb{C}^N\otimes\mathbb{C}^2$, and associate to $A$ a unitary quantum walk in this expanded space, $\mathbb{C}^{2N}\otimes \mathbb{C}^{2N}$. The goal is to construct a method that implements a Chebyshev polynomial function $\T_n(\A)$, but the construction below will instead produce $\T_n(\A/d)$.

We start by considering an N-dimensional hyperplane $V \subset \mathbb{C}^{2N}\otimes \mathbb{C}^{2N}$, spanned by the states 
\begin{equation}
\ket{\psi_j}=\ket{j}\otimes \frac{1}{\sqrt{d}} \sum_{k:\A_{jk}\neq 0}\left(\sqrt{\A_{jk}^*}\ket{k}+\sqrt{1-\Abs{\A_{jk}}}\ket{k+N}\right),
\end{equation}
for $j=1,\ldots,N$. These states are orthonormal, since we can assume that for rows with fewer than $d$ entries, the summation is taken over enough zero entries to make up the total of $d$ terms in the linear combination. Next, define an isometry $T:\mathbb{C}^{N}\rightarrow\mathbb{C}^{2N}\otimes \mathbb{C}^{2N}$ that embeds the system register into the expanded space
\begin{equation}
T := \sum_{i=1}^N \ket{\psi_j}\bra{j}.
\end{equation}
The third ingredient is a swap operator on $\mathbb{C}^{2N}\otimes \mathbb{C}^{2N}$
\begin{equation}
S\ket{j,k}:=\ket{k,j} \hfill \forall j,k = 1,\ldots,2N.
\end{equation}
The unitary walk operator is then defined as
\begin{equation}
\W:=S\left(2TT^{\dagger}-\id\right).
\end{equation}

Now it can be shown that in a subspace $\mathcal{B}\subset\mathbb{C}^{2N}\otimes \mathbb{C}^{2N}$, defined as the span of the isometric mappings of the basis states $\ket{j}$, and the corresponding swapped states, the walk operator has a block form such that the blocks are Chebyshev polynomials in $\A/d$, i.e. $\T_n(\A/d)$ and $U_n(A/d)$.

\begin{lemma}
\label{lem:Wblock}
The $2N$-dimensional subspace $\mathcal{B}=\textnormal{span}\{T\ket{j}, ST\ket{j} |\ j=1,\ldots,N\} \subset \mathbb{C}^{2N}\otimes \mathbb{C}^{2N}$ is invariant under the walk operator $\W$. Further, $\W$ can be put in a block form on $\mathcal{B}$
\begin{equation}
\label{eqn:Wblock}
\W|_{\mathcal{B}}=\begin{bmatrix}
	H & -\sqrt{1-H^2} \\
    \sqrt{1-H^2} & H
\end{bmatrix},
\end{equation}
where $H=\A/d$. We henceforth drop the subscript $\mathcal{B}$, and work only in this subspace.
\end{lemma}

\begin{lemma}
\label{lem:walk_n}
For a matrix 

\begin{equation}
\W=\begin{bmatrix}
	\lambda & -\sqrt{1-\lambda^2} \\
    \sqrt{1-\lambda^2} & \lambda
\end{bmatrix}
\end{equation}
 
where $\Abs{\lambda}\leq 1$, and any positive integer $n$,
\begin{equation}
\W^n=\begin{bmatrix}
	\T_n(\lambda) & -\sqrt{1-\lambda^2}U_{n-1}(\lambda) \\
    \sqrt{1-\lambda^2}U_{n-1}(\lambda) & \T_n(\lambda)
\end{bmatrix},
\end{equation}
where $\T_n(x)$ and $U_n(x)$ are the Chebyshev polynomials of the first and second kinds respectively, having degree $n$, defined on $[-1,1]$.
\end{lemma}

Since $H$ is Hermitian, its eigenvectors $\ket{\lambda}$ span $\mathbb{C}^N$. Thus within the invariant subspace $\mathcal{B}$ of $\W$, for any state $\ket{\psi}\in \mathbb{C}^N$, we combine the above two lemmas to get
\begin{equation}
\W^n T \ket{\psi}\rightarrow T\T_n(H)\ket{\psi}+\ket{\psi^{\perp}},
\end{equation}
where $\ket{\psi^{\perp}}$ is orthogonal to $T\ket{j}$ for each $j=1,\ldots,N$, but is not normalised.  

Since $T$ is an isometry, we can dilate and implement it by a unitary circuit that performs the map $\ket{0^m}\ket{\psi}\mapsto T\ket{\psi}$ for any state $\ket{\psi}\in \mathbb{C}^N$, with $m=\left\lceil\log 2N\right\rceil+1$ ancillaries. Hence applying $T^{\dagger}\W^n T$ will perform the map 
\begin{equation}
\label{eqn:walk_cheb}
\ket{0^m}\ket{\psi}\mapsto \ket{0^m}\T_n(H)\ket{\psi}+\ket{\Phi^{\perp}},
\end{equation}
where $\ket{\Phi^{\perp}}$ is not normalised but $\left(\ket{0^m}\bra{0^m}\otimes\id_N\right)\ket{\Phi^{\perp}}=0$. Post-selecting on measuring the first $m$ registers to be in the `$0$' state, we get a probabilistic implementation of the function $\T_n(H)$ with $H=\A/d$. 

We need only $\O(1)$ queries to the oracle $\mathcal{P}_\A$ to implement the  walk operator $\W$ as well as the isometry $T$ \cite{Berry2012}. So taking $n$ steps of the walk requires $\O(n)$ queries to $\mathcal{P}_\A$.

\end{document}